\documentclass[letterpaper, 10 pt, conference]{ieeeconf}  

\IEEEoverridecommandlockouts                              

\usepackage{amsmath,amsthm,amssymb,amsfonts}
\usepackage{xcolor}
\usepackage{graphicx}
\usepackage{subcaption}
\usepackage[american]{circuitikz}
\usepackage[thinc]{esdiff}
\usepackage{tabularx}
\usepackage{multirow}
\usepackage{booktabs}
\usepackage[UKenglish]{babel}
\usepackage[lined]{algorithm2e}
\newtheorem{remark}{Remark}[section]
\newtheorem{lemma}{Lemma}[section]
\newtheorem{defn}{Definition}[section]
\newtheorem{thm}{Theorem}[section]
\newtheorem{prop}{Proposition}[section]
\newtheorem{asu}{Assumption}[section]

\title{\vspace{-13mm}\LARGE \bf
	Moment Matching for Descriptor Systems:\\ A M\" obius Mapping Approach\\{\color{white}(The Sandsnakes' Version)}\vspace{-5mm}
}

\author{David Militaru$^{1}$, Denisa Jungheatu$^{1,*}$, Andrei Speril\u a$^{1}$, and Tudor C. Ionescu$^{1,2}$\vspace{-10mm}
	\thanks{$^{1}$The authors are with the Department of Automatic Control and Systems Engineering, ``Politehnica'' University of Bucharest, Bucharest, Romania
		{\tt\footnotesize \{denisa.jungheatu,david.militaru\}@stud.acs.upb.ro, \{andrei.sperila,tudor.ionescu\}@upb.ro}}%
    \thanks{$^{2}$Tudor C. Ionescu is also a member of the ``Gheorghe Mihoc-Caius Iacob'' Institute of Mathematical Statistics and Applied Mathematics of the Romanian Academy, Bucharest, Romania}
	\thanks{$^{*}$Denisa Jungheatu is the corresponding author. Please address all electronic correspondence to {\tt\footnotesize denisa.jungheatu@stud.acs.upb.ro}.}%
}

\begin{document}

	\maketitle
	\thispagestyle{empty}
	\pagestyle{empty}

	\begin{abstract}
		
    For a class of single-input single-output systems described by proper or improper transfer functions, we propose moment-matching procedures applicable in both continuous- and discrete-time contexts. The resulting technique is not only more flexible and reliable than other procedures which are currently available in literature, but it also enables the placement of constraints on the reduced-order model's poles and zeros. This constraint-based feature, hitherto available only for continuous-time state-space systems, is illustrated via a numerical example based on a practical problem from literature.\vspace{-1mm}
		
	\end{abstract}

	\section{Introduction}\label{sec:intro}\vspace{-1mm}
	
	\subsection{Context}\label{subsec:context}
	
	The moment-matching framework, originally proposed in \cite{ORIG_MM} for both linear and nonlinear systems, is one of the most popular approaches to model reduction. This framework has drawn considerable interest over the years, with numerous subsidiary techniques having been developed (primarily in the continuous-time setting) with the aim of matching a system's steady-state response (in both the linear \cite{TS_MM} and the nonlinear \cite{NL_MM} case), or of optimising an $\mathcal{H}_2$-norm-based approximation error \cite{H2_MM}. Recently, however, this class of techniques has received renewed attention due to their affinity towards data-driven solutions (see, for example, \cite{IO_DD_MM} and \cite{KBL_MM}), which has also spurred interest in constructing discrete-time analogues of the aforementioned procedures \cite{DISC_MM}.
	
	Some of the modern technical applications that benefit the most from model approximation are large-scale networked systems, which oftentimes exhibit singular dynamics (see \cite{DAE_surv} for a comprehensive discussion). To conveniently tackle the latter, a combination of centred realisations \cite{css_orig} and M\" obius mappings (see Chapter~10 in \cite{complexvar}) was shown in \cite{H2_ct} and in \cite{H2_ds} to elegantly bypass all of these difficulties, by allowing a direct analogue of classical theory to be employed in both continuous- and discrete-time cases. This type of approach, however, seems to be generally under-represented in the field of model approximation, with recent publications grappling directly with a network's original, singular dynamics \cite{DAE_MM}. \vspace{-2mm}
	
	\subsection{Motivation}\label{subsec:motiv}
	
	Although the technique presented in \cite{DAE_MM} enables moment matching for singular systems, we point out the fact that the procedure is applicable only in the continuous-time setting, and that it does not permit the placement of the approximated model's poles or zeros. Additionally, this approach employs the so-called descriptor realisations, which exhibit a series of computational drawbacks compared to other representations that will be discussed in the sequel. The same limitation regarding pole-zero placement also applies to \cite{DISC_MM}, as well as the overall inability of the procedure to handle both continuous- and discrete-time systems via a unified framework.\vspace{-2mm}
	
	\subsection{Contribution}\label{subsec:contrib}
	
	The current work hereby proposes a moment-matching-based approximation technique, inspired chiefly by the results presented in \cite{PZ_MM}, and applicable to both continuous- and discrete-time systems. For single-input single-output plants with an associated transfer function that may be improper, our proposed model-reduction procedure allows us to select the poles and zeros of the approximated system. Furthermore, due to the use of specialised representations, our moment matching routines are guaranteed to be more flexible and reliable than the ones in \cite{DAE_MM} whenever the system's associated transfer function is improper and, therefore, said transfer function cannot be represented via state-space systems (for which classical moment matching theory can be applied).\vspace{-2mm}
	
	\subsection{Paper Structure}\label{subsec:struc}
	
	The rest of the manuscript is organised as follows: Section~\ref{sec:prelim} presents a number of preliminary notions and auxiliary results. Section~\ref{sec:main} holds our problem statement and our main results, which are then illustrated on a practical problem in Section~\ref{sec:num_ex}. Concluding remarks are presented in Section~\ref{sec:outro}.\vspace{-1mm}
	
	\section{Preliminaries}\label{sec:prelim}\vspace{-1mm}
	
	\subsection{Notation}\label{subsec:not}
	
	Let $\mathbb{N}$, $\mathbb{R}$ and $\mathbb{C}$ be the set of all natural, real, and complex numbers. The one-point compactification of $\mathbb{C}$ (Riemann's Sphere) will be denoted $\overline{\mathbb{C}}:=\mathbb{C}\cup\{\infty\}$. The set of $p\times m$ matrices with entries in a set $\mathbb{M}$ is denoted $\mathbb{M}^{p\times m}$, and $\mathbb{M}^p$ is the set of vectors of length $p$ having entries in $\mathbb{M}$. The term $\mathrm{Jord}(\mu_1,n_1,\dots,\mu_q,n_q)$ denotes a matrix in Jordan canonical form, in which the Jordan block associated with an eigenvalue $\mu_i\in\mathbb{C}$ is of size $n_i\times n_i,\,\forall\,i\in\{1:q\}$. For any $Q\in\mathbb{C}^{p\times m}$, the matrices $\mathrm{Re}(Q)$ and $\mathrm{Im}(Q)$ are formed from the real and imaginary parts of the entries belonging to $Q$, whereas the term $\|Q\|_*$ denotes the \emph{nuclear norm} of $Q$, which represents the sum of all its singular values.
	
	For any two matrices $M, N\in\mathbb{C}^{p\times m}$, the first-order matrix polynomial $M-\lambda N$ of complex indeterminate $\lambda$ is called a \emph{matrix pencil}, and we denote its set of \emph{generalised eigenvalues} (see \cite{gantmacher}) as $\Lambda(M-\lambda N)$. For any matrix $M$, we denote by $M^\top$ its transpose and, if $M=M^\top$ is real-valued, then $M\succ O$ will be used to denote the fact that $M$ is positive definite. Additionally, all real-rational functions of complex indeterminate $\lambda$ (which stands for either the Laplace-transform variable $s$ in continuous time, or the $\mathcal{Z}$-transform variable $z$ in discrete time) will be termed \emph{transfer functions} and will be denoted via boldface symbols.\vspace{-1mm}
	
	\subsection{Descriptor Systems and Centred Realisations}\label{subsec:dss}
	
	Our focus in this manuscript is on systems whose input-output dynamics are described in the frequency domain by\vspace{-2mm}
	\begin{equation}\label{eq:tf}
		\mathbf{G}(\lambda)=\frac{b_{\mathbf{G}}(\lambda)}{a_{\mathbf{G}}(\lambda)}=\frac{\sum_{j=0}^{n_b}b_j\lambda^j}{\sum_{i=0}^{n_a}a_i\lambda^i}\,,\vspace{-1mm}
	\end{equation}
	\emph{transfer functions} with $a_i,b_j\in\mathbb{R},\,\forall\,i\hspace{-0.5mm}\in\hspace{-0.5mm}\{0\hspace{-0.5mm}:\hspace{-0.5mm}n_a\},\ j\hspace{-0.5mm}\in\hspace{-0.5mm}\{0\hspace{-0.5mm}:\hspace{-0.5mm}n_b\}$, $a_{n_a}b_{n_b}\neq0$, and $n_a\in\mathbb{N}$ \emph{possibly smaller} than $n_b\in\mathbb{N}$.\vspace{-1mm}
	
	\begin{defn}
		Let $\mathbf{G}(\lambda)$ be expressed as in \eqref{eq:tf}, in which the pair of polynomials $(a_{\mathbf{G}}(\lambda),b_{\mathbf{G}}(\lambda))$ is coprime. The \emph{McMillan degree} of $\mathbf{G}(\lambda)$ is defined as $\delta(\mathbf{G}):=\max\{n_{a},n_b\}$ (see also Section~3.11 of \cite{zhou} for an alternative interpretation), whereas the roots of the coprime polynomials $a_{\mathbf{G}}(\lambda)$ and $b_{\mathbf{G}}(\lambda)$ are, respectively, the \emph{finite} poles and zeros of $\mathbf{G}(\lambda)$.\vspace{-1mm}
	\end{defn}
	
	\begin{defn}\label{def:inf_pz}
		The \emph{infinite} poles and zeros of $\mathbf{G}(\lambda)$ are defined as the finite poles and zeros of $\mathbf{G}\left(\frac{1}{\lambda}\right)$ at $\lambda=0$.\vspace{-1mm}
	\end{defn}
	
	Given a point $\lambda_0\in\overline{\mathbb{C}}$ and a real-valued matrix quintuplet, we denote by $(A,B,C,D,E)_{\lambda_0}$ a so-called \emph{realisation centred at} $\lambda_0$ of $\mathbf{G}(\lambda)$. When $\lambda_0\in\mathbb{C}$, the realisation satisfies\vspace{-1mm}
	\begin{equation}\label{eq:css}
		\hspace{-3mm}\footnotesize\begin{array}{l}
		      \mathbf{G}(\lambda)=C(\lambda E-A)^{-1}B(\lambda_0-\lambda)+D=:\hspace{-1mm}\left[\begin{array}{c|c}
		    A-\lambda E & B \\\hline
		      C & D
		\end{array}\right]_{\lambda_0}\hspace{-1mm},
		\end{array}\vspace{-2mm}\hspace{-4mm}\normalsize
	\end{equation}
    and, when $\{\lambda_0\}=\{\infty\}$, we denote\vspace{-2mm}
    \begin{equation*}
		\hspace{-3mm}\footnotesize\begin{array}{l}
		      \mathbf{G}(\lambda)=C(\lambda E-A)^{-1}B+D=:\hspace{-1mm}\left[\begin{array}{c|c}
		    A-\lambda E & B \\\hline
		      C & D
		\end{array}\right].
		\end{array}\vspace{-1mm}\normalsize
	\end{equation*}
    Note that, by dropping the $\lambda_0$ index, we get $(A,B,C,D,E)$, which is precisely a so-called descriptor realisation of the type used in \cite{DAE_MM}. Moreover, a descriptor realisation having $E=I_n$ is nothing more than a standard \emph{state-space system}, with the employed notation further reducing to $(A,B,C,D)$. In all of these cases, we have $A,E\in\mathbb{R}^{n\times n}$, $B,C^\top\in\mathbb{R}^n$, and $D\in\mathbb{R}$, with $n\in\mathbb{N}$ being called the realisation's \emph{order} and the matrix polynomial $A-\lambda E$, which implicitly satisfies $\det(A-\lambda E)\not\equiv 0$, being called the realisation's \emph{pole pencil}.\vspace{-1mm}

    \begin{remark}
		It is precisely due to the fact that centred realisations generalise the notion of descriptor realisation that we can derive the many theoretical and computational benefits presented in the sequel. Note also that while our representations are time-domain-independent, those used in \cite{DAE_MM} are tailored exclusively to the continuous-time case.\vspace{-1mm}
	\end{remark}

    \noindent We now introduce a crucial property of centred realisations.\vspace{-1mm}
	
	\begin{defn}
		A centred realisation is called \emph{minimal} if its order is no greater than that of any other centred realisation of the same transfer function, regardless of their respective centring points.\vspace{-1mm}
	\end{defn}

	The next result states a succinct condition for minimality.\vspace{-1mm}
	
	\begin{prop}\label{prop:min_css}
		For any $\mathbf{G}(\lambda)$, a centred realisation is minimal if and only if its order is equal to $\delta(\mathbf{G})$.\vspace{-2mm}
	\end{prop}
	\begin{proof}
		See the proof of point~iii) in Lemma~2 from \cite{H2_ds}.\vspace{-1mm}
	\end{proof}

    The following algebraic property is intrinsically tied to the concept of minimality (see Chapter~3 of \cite{zhou} for details on this connection) and will be used extensively in the sequel.\vspace{-1mm}

    \begin{defn}
		A matrix pair $(A,B)$, with $A\in\mathbb{R}^{n\times n}$ and $B\in\mathbb{R}^{n\times m}$, is called \emph{controllable} if\vspace{-2mm}
		\begin{equation*}
			{\mathrm{rank}\begin{bmatrix}
					B & AB & \dots & A^{n-1}B
				\end{bmatrix}=n,}\vspace{-2mm}
		\end{equation*}
		and $\left(B^\top,A^\top\right)$ is called \emph{observable} if $(A,B)$ is controllable.\vspace{-1.5mm}
	\end{defn}
	
	Finally, we proceed to define the concept of moment from a purely frequency-based perspective, just as in \cite{ORIG_MM}.\vspace{-1mm}
    
    \begin{defn}\label{def:mom}
         For any point $\mu\in\overline{\mathbb{C}}$ that is not a pole of $\mathbf{G}(\lambda)$, the \emph{$i^\text{th}$-order moment} of $\mathbf{G}(\lambda)$ at $\mu$ is defined as\vspace{-1mm}
    	\begin{equation}\label{eq:moment_def}
    		\hspace{-1mm}\eta_{\,\mathbf{G},i}(\mu):= \frac{(-1)^i}{i!}\mathbf{G}^{(i)}(\mu) =\frac{(-1)^i}{i!}\hspace{-1mm}\left.\left(\frac{\mathrm{d}^i}{\mathrm{d}\lambda^i}\mathbf{G}(\lambda)\right)\right\vert_{\lambda=\mu}.\vspace{-2mm}\hspace{-2mm}
    	\end{equation}
    \end{defn}
	
	\subsection{M\" obius Mappings}\label{subsec:FST}
	
	We briefly present here the class of functions discussed at length in Chapter~10 of \cite{complexvar}. Let $\mathfrak{g}:\overline{\mathbb{C}}\mapsto\overline{\mathbb{C}}$ be defined as\vspace{-1mm}
	\begin{equation}\label{eq:moeb}
		\mathfrak{g}(\lambda):=\frac{a\lambda+b}{c\lambda+d}\,,\vspace{-1mm}
	\end{equation}
	for any $a,b,c,d\in\mathbb{C}$ such that $ad-bc\neq0$. This choice of parameters yields the following crucial property.\vspace{-2mm}
	\begin{lemma}\label{lem:inv}
		The function defined as in \eqref{eq:moeb} with $ad-bc\neq0$ is bijective, and its inverse is given by $\mathfrak{g}^{-1}:\overline{\mathbb{C}}\mapsto\overline{\mathbb{C}}$ with\vspace{-1mm}
		\begin{equation*}
			\mathfrak{g}^{-1}(\lambda)=\frac{\phantom{-}d\lambda-b}{-c\lambda+a}\,.\vspace{-1mm}
		\end{equation*}
	\end{lemma}
	\begin{proof}
		See Chapter~10 of \cite{complexvar}.\vspace{-2mm}
	\end{proof}
    Given a finite point $\lambda_0\in\mathbb{C}$, a particular instance of the mappings from \eqref{eq:moeb}, denoted $\mathfrak{f}:\overline{\mathbb{C}}\mapsto\overline{\mathbb{C}}$ and defined as\vspace{-1mm}
    \begin{equation}\label{eq:f_def}
        \mathfrak{f}(\lambda):=\frac{1}{\lambda - \lambda_0},\vspace{-1mm}
    \end{equation}
	displays remarkable properties with respect to realisations of type \eqref{eq:css}, as shown by the following result.\vspace{-2mm}
	\begin{lemma}\label{lem:remap}
		Let a transfer function $\mathbf{G}(\lambda)$ be realised as in \eqref{eq:css} with a $\lambda_0\in\mathbb{R}$ that ensures $A-\lambda_0 E$ is invertible. Then, by taking the mapping $\mathfrak{f}(\lambda)$ from \eqref{eq:f_def}, we have that:
		\begin{enumerate}
			\item[a)] The inverse of $\mathfrak{f}(\lambda)$ is given by $\mathfrak{f}^{-1}(\lambda)=\lambda_0+\frac{1}{\lambda}$;\vspace{0.5mm}
				
			\item[b)] The transfer function $\widetilde{\mathbf{G}}(\lambda):=\mathbf{G}\left(\mathfrak{f}^{-1}(\lambda)\right)$ can be realised by $((A-\lambda_0E)^{-1}E,(A-\lambda_0E)^{-1}B,C,D)$;\vspace{0.5mm}
			
			\item[c)] If the realisation from \eqref{eq:css} is minimal, then so is the one from point~\emph{b)} of this result.\vspace{-2mm}
		\end{enumerate}
	\end{lemma}
	\begin{proof}
		Point a) follows from Lemma~\ref{lem:inv}, whereas the proof of points b) and c) is exactly that of Lemma~2 in \cite{H2_ct}. \vspace{-1mm}
	\end{proof}
    
	\begin{remark}
		The restriction of $\lambda_0$ to $\mathbb{R}$ in Lemma~\ref{lem:remap} is done purely to preserve the real-valued nature of the matrices that form the realisation expressed in point \emph{a)} of said result. Moreover, the freedom in choosing $\lambda_0\in\mathbb{R}$ ensures that the centring point can \emph{always} be chosen in order to ensure that $A-\lambda_0 E$ is well-conditioned with respect to matrix inversion.\vspace{-3mm}
	\end{remark}
	
	\begin{figure*}
			\begin{equation}\label{eq:zero_cond}\tag{6}
				\small\begin{array}{cc}
					z_\ell\in\mathbb{C}\setminus\{\lambda_0\}\Rightarrow z_\ell\in\Lambda\left( \begin{bmatrix}
						A_m - \lambda E_m &\hspace{-1mm} B_m(\lambda_0-\lambda)\\
						C_m & D_m
					\end{bmatrix} \right),& \begin{array}{l}
						\{z_\ell\}=\{\infty\}\Rightarrow \begin{bmatrix}
							A_m - \lambda E_m &\hspace{-1mm} B_m(\lambda_0-\lambda)\\
							C_m & D_m
						\end{bmatrix}\text{ has an infinite}\\
						\,\text{generalised eigenvalue with a partial multiplicity greater than }1.\vspace{-2mm}
					\end{array}\hspace{-3mm}
				\end{array}\normalsize
			\end{equation}\vspace{-6mm}
		\hrulefill
	\end{figure*}
	
	\section{Main Results}\label{sec:main}\vspace{-1mm}
	
	\subsection{Problem Statement}\label{subsec:prob_st}
	
	Let $\mathbf{G}(\lambda)$ { be described by a minimal centred realisation denoted} $(A,B,C,D,E)_{\lambda_0}$ of order $n=\delta(\mathbf{G)}$, and for which the scalar $\lambda_0\in\mathbb{R}$ is not a pole of $\mathbf{G}(\lambda)$. The objective is to find a second transfer function $\mathbf{H}(\lambda)\not\equiv0$, described by a centred realisation of order $\nu>0$ with $\delta(\mathbf{H})\leq\nu<n$ and denoted $(A_m,B_m,C_m,D_m,E_m)_{\lambda_0}$, which, for a set of \emph{distinct} points in $\overline{\mathbb{C}}\setminus\{\lambda_0\}$ denoted ${\mu_i\not\in\Lambda(A-\lambda E)}$, $p_k$, $z_\ell\not\in\Lambda(A-\lambda E)$, and the integers $n_i>0$, $\alpha>0$, $\beta\geq 0$, $\gamma\geq 0$ which satisfy $ \sum_{i=1}^{\alpha} n_i = \nu $ and $\beta+\gamma \leq \nu$, achieves:
	\begin{enumerate}
		\item[A)] $\eta_{\,\mathbf{G},j}({\mu_i}) = \eta_{\,\mathbf{H},j}({\mu_i}),\,\forall\, i \in \{1:\alpha \},\, j \in \{0:n_i-1\}$;
		\item[B)] $p_k \in \Lambda \left(A_m - \lambda E_m \right), \, \forall\, k \in\{1:\beta\}$;
        \item[C)] $z_\ell$ satisfies \eqref{eq:zero_cond}\stepcounter{equation}, at the top of the next page, $\forall\ell\hspace{-0.25mm}\in\hspace{-0.25mm}\{1:\gamma\}$.\vspace{-1mm}
	\end{enumerate}
	
	\begin{remark}\label{rem:obj_BC}
		It is important to point out that, when $\beta = 0$ or $\gamma = 0$, the corresponding constraints from objectives \emph{B)} or \emph{C)} will be disregarded when constructing $\mathbf{H}(\lambda)$. Note also that these objectives are analogues to classical state-space results (see Chapter~3 of \cite{zhou}, along with \cite{gen_ss} for the infinite zero condition) used to impose that $p_k$ and $z_\ell$ are poles and, respectively, zeros \emph{of the reduced-order model's realisation}.\vspace{-1mm}
	\end{remark}
	
	\subsection{Reduced-order Models for Moment Matching}\label{subsec:fam}
	
	We begin by addressing objective A) stated in the previous section and, in doing so, we hereby provide \emph{an entire class} of reduced-order models achieving moment matching at the desired set of points, along with the corresponding orders. Before doing so, we make the following assumption.\vspace{-1mm}

    \begin{asu}\label{asu:not_id_zero}
        When imposing the pairs $(\mu_i,n_i)$ for moment matching, there exist two indices $i\in\{1:\alpha\}$ and $j\in\{0:n_i-1\}$ such that $\eta_{\,\mathbf{G},j}(\mu_i)\neq 0$.\vspace{-1mm}
    \end{asu}

    \begin{remark}\label{rem:not_id_0}
        Assumption~\ref{asu:not_id_zero} is by no means restrictive, from a practical standpoint, since moment matching habitually seeks to preserve non-zero moments of the system in the complex plane. The reason for this assumption is to ensure that, for all $\mathbf{H}(\lambda)$ satisfying objective \emph{A)}, we have $\mathbf{H}(\lambda)\not\equiv 0$.\vspace{-1mm}
    \end{remark}
    
    The following theorem, which constitutes our main result, provides an appealing parametrisation.\vspace{-1mm}
	
	\begin{thm}\label{thm:class}
		Let $\mathbf{G}(\lambda)$ be given as in Section~\ref{subsec:prob_st}. For\newline any positive integer $\nu<n$, consider any $S\in\mathbb{R}^{\nu\times \nu}$ which is similar to the matrix $J:=\mathrm{Jord}(\mathfrak{f}(\mu_1), n_1, \dots, \mathfrak{f}(\mu_\alpha),n_{\alpha})$ and any row vector $L\in\mathbb{R}^{1\times \nu}$ such that $(L,S)$ is observable. Then, the following two statements hold:
		\begin{enumerate}
			\item[a)] There exists $\Pi\in\mathbb{R}^{n\times \nu}$, which is the unique solution of the generalised Sylvester equation given by\vspace{-1mm}
			\begin{equation} \label{eq:gen_sylv}
				E \Pi + BL = (A-\lambda_0 E)\Pi S\,;\vspace{-2mm}
			\end{equation}
			
			\item[b)] For any column vector $G\in \mathbb{R}^{\nu}$ which ensures that\vspace{-1mm}
			\begin{equation}\label{eq:not_pole}
				\mathfrak{f}(\mu_i)\not\in\Lambda(S-GL-\lambda I_\nu),\,\forall\,i\in\{1:\alpha\},\vspace{-1mm}
			\end{equation}	
			the following transfer function\vspace{1mm}
            \end{enumerate}
			\begin{equation} \label{eq:has}
				\small\begin{array}{c}
				    \mathbf{H}(\lambda) =
                    \left[\begin{array}{c|c}
                        I_{\nu} + \lambda_0 (S-GL) -\lambda (S-GL) & G \\\hline
                        C\Pi & D
                    \end{array}\right]_{\lambda_0}
                    \end{array}\vspace{-1mm}\hspace{-1mm}\normalsize
			\end{equation}
            \begin{enumerate}
			\item[]is well-defined and fulfils objective \emph{A)} {of Section~\ref{subsec:prob_st}.}\vspace{-2mm}
		\end{enumerate}
	\end{thm}
	\begin{proof}
		See the Appendix.\vspace{-1mm}
	\end{proof}
	
	The following remark highlights the benefits of employing our approach, compared to the one stated in \cite{DAE_MM}.\vspace{-2mm}
	\begin{remark}
		In contrast to the solutions from Corollaries~1 and~2 in \cite{DAE_MM}, note that each member in the class of reduced-order models from Theorem~\ref{thm:class} has a pole pencil that is \emph{guaranteed} to be invertible, ensuring that \emph{all} our reduced-order models have well-defined transfer functions. Moreover, note that the chosen order $\nu$ also acts as a \emph{tight upper bound} on the McMillan degree of these transfer functions, which is \emph{not always} the case when working with descriptor systems whose transfer functions have poles at infinity (see \cite{gen_ss}).\vspace{-2mm}
	\end{remark}
	
	A further benefit of our technique is the fact that it enables the placement of constraints on the poles and the zeros of\newpage\noindent the reduced-order model's transfer function. This feature and its implementation form the core topic of our next subsection.\vspace{-2mm}
	
	\subsection{Pole and Zero Constraints}\label{subsec:pz_con}
	
	We now proceed to tackle the final two objectives stated in Section~\ref{subsec:prob_st}, by extending and adapting the propositions given in Section~3 of \cite{PZ_MM}. Yet, before doing so, we clarify our particular use of terminology via the following remark.\vspace{-2mm}
	
	\begin{remark}\label{rem:inv_pencil}
		Although objectives \emph{B)} and \emph{C)} from Section~\ref{subsec:prob_st} refer only to the centred realisation of the reduced-order model, note that said constraints \emph{almost surely} (in a statistical sense) impact said model's transfer function. Since the properties of controllability and observability are \emph{generic} (see, for example, \cite{Won}), it follows that the points denoted $p_k$ and $z_\ell$ (recall Remark~\ref{rem:obj_BC}) will \emph{almost surely} be poles and zeros of the reduced-order model's transfer function.\vspace{-2mm}
	\end{remark}
	
	We may now state the following result, which tackles the problem of placing the poles of the reduced-order model.\vspace{-2mm}
	
	\begin{prop}\label{prop:poles}
		Under the same framework of notation and assumptions made in Theorem~\ref{thm:class}, define the matrices $D_k:=\mathfrak{f}(p_k)I_\nu-S,\,\forall\,k\in\{1:\beta\}$. Then, each such $D_k$ is nonsingular and, by defining the matrices $\mathcal R_k:=\mathrm{Re}(D_k^{-1})$ along with $\mathcal I_k:=\mathrm{Im}(D_k^{-1})$, if $G\in \mathbb{R}^{\nu}$ satisfies \eqref{eq:not_pole} and\vspace{-1mm}
		\begin{equation}\label{eq:pole_cond}
			\small\begin{bmatrix}
			    1\\0
			\end{bmatrix}+\small\begin{bmatrix}
			    L\mathcal R_k\\L\mathcal I_k
			\end{bmatrix}G=\small\begin{bmatrix}
			    0\\0
			\end{bmatrix},\,\forall\,k\in\{1:\beta\},\vspace{-1mm}\normalsize
		\end{equation}
		it follows that the centred realisation given in \eqref{eq:has} also satisfies objective \emph{B)} from Section~\ref{subsec:prob_st}.\vspace{-2mm}
	\end{prop}
	\begin{proof}
		See the Appendix.\vspace{-2mm}
	\end{proof}

    \begin{remark}
        Note that, from a numerical standpoint, the matrix pairs $(\mathcal R_k, \mathcal I_k)$ can be computed before tackling the linear constraints from \eqref{eq:pole_cond}, which enables the checking of these constraints' numerical conditioning. As expected, when placing poles in the reduced-order model that are close to the locations in $\overline{\mathbb{C}}$ where the original model's moments are being matched, numerical accuracy naturally becomes a crucial and sensitive aspect of the reduction process.\vspace{-2mm}
    \end{remark}
	
	A more elaborate result than the one above can be derived when attempting to impose the reduced-order model's zeros.\vspace{-2mm}
    
	\begin{thm}\label{thm:zeros}
		
		Under the same framework of notation and assumptions made in Theorem~\ref{thm:class}, consider any matrix $S_z$ which is similar to $J_z:=\mathrm{diag}(\mathfrak{f}(z_1),\dots,\mathfrak{f}(z_\gamma))$, along with any vector $R_z\in\mathbb{R}^{\gamma}$ such that $(S_z,R_z)$ is controllable. Then:
		\begin{enumerate}
			\item[a)] There exists $\widehat{\Pi}\in\mathbb{R}^{\nu\times n}$ such that $\widehat{\Pi}\Pi=I_\nu$ and, for any such $\widehat{\Pi}$, the matrix $\Upsilon_z\in\mathbb{R}^{\gamma\times n}$ is the unique solution of the generalised Sylvester equation given by\vspace{-2mm}
			\begin{equation}\label{eq:sylv_z}
				S_z\Upsilon_z(A-\lambda_0 E) = \Upsilon_z E + R_z \big(C+DL\widehat{\Pi}\big)\,;\vspace{-2mm}
			\end{equation}
			
			\item[b)] If $\mathrm{rank}(\Upsilon_z(A-\lambda_0 E)\Pi) = \gamma$ and $\Upsilon_zB=O$, whereas $G\in \mathbb{R}^{\nu}$ satisfies \eqref{eq:not_pole}  along with\vspace{-2mm}
			\begin{equation}\label{eq:suff_sing}
			\Upsilon_z(A-\lambda_0 E)\Pi G + R_zD = O,\vspace{-2mm}
			\end{equation}
				then the centred realisation given in \eqref{eq:has} also satisfies objective \emph{C)} from Section~\ref{subsec:prob_st}.\vspace{-1mm}
		\end{enumerate}
	\end{thm}
	\begin{proof}
		See the Appendix.\vspace{-2mm}
	\end{proof}

    \begin{figure*}
        \begin{equation}\label{eq:prob_opt}\tag{14}
            \mathcal{P}(X,Y,\Theta):=\left\{\begin{array}{l}
                \min\limits_{W,\widehat{W},G,\widehat{\Pi},P_{W},P_{G}}\left\|\small\begin{bmatrix}
                    T_C + XPY + T_APY + XPT_B & (T_A+X)P\\
                    P(T_B+Y) & P
                \end{bmatrix}\right\|_*,\\\vspace{-3mm}\\
                \text{\phantom{aa}subject to }\widehat{\Pi}\Pi=I_\nu,\ \eqref{eq:sylv_W}\text{ holds},\ P_W\succ O,\ P_G + R_zD = O,\ \Theta = O.
            \end{array}\right.\vspace{-1mm}
        \end{equation}
        \hrulefill\vspace{-6mm}
    \end{figure*}

    Before showcasing the efficiency of our technique on a practical problem inspired by the electric circuits from \cite{app}, we first discuss the numerical considerations of satisfying the algebraic conditions derived in this section of the paper.\vspace{-1mm}
	  
	\subsection{Numerical Considerations}\label{subsec:num_cons}

    We begin by pointing out that point a) of Theorem~\ref{thm:class} and the entire statement of Proposition~\ref{prop:poles} revolve around linear equality constraints, in which $\Pi$ and $G$ are the variable terms, respectively. Indeed, once a solution for the generalised Sylvester equation in \eqref{eq:gen_sylv} is found (either by a dedicated linear algebra solver or by an optimisation toolbox computing a feasible solution for a linear programming problem) and the class of systems in \eqref{eq:has} is parametrised solely in $G$,  note that pole placement for the reduced-order model amounts to merely solving the system of linear equations in \eqref{eq:pole_cond}.
    
    If more than one solution exists  to this problem, then we can further decompose the gain matrix as $G=G_0+Fx$, where $G_0$ is any vector in $\mathbb{R}^\nu$ satisfying \eqref{eq:pole_cond}, $F\in\mathbb{R}^{\nu\times q}$ is chosen as being any basis matrix for the right nullspace of $\begin{bmatrix}
        (L\mathcal R_1)^\top & (L\mathcal I_1)^\top & \dots & (L\mathcal R_\beta)^\top & (L\mathcal I_\beta)^\top
    \end{bmatrix}^\top$ (having denoted the dimension of this space as $q$), and $x$ is arbitrary in $\mathbb{R}^q$. Replacing the newly given expression of $G$ in \eqref{eq:has}, and subsequently parametrising in terms of $x$, will yield a class of reduced-order models with the desired pole structure.

    On the other hand, the more challenging aspect of the technique proposed in this paper is \emph{simultaneously} satisfying the conditions in Theorem~\ref{thm:zeros}. The identity $\Upsilon_zB=O$ can be tackled by rewriting $\Upsilon_z=(F_BW)^\top$, where $F_B\in\mathbb{R}^{n\times q_B}$ is chosen as being any basis matrix for the right nullspace of $B^\top$ (having denoted the dimension of this space as $q_B$) and $W$ is arbitrary in $\mathbb{R}^{q_B\times \gamma}$, thus turning \eqref{eq:sylv_z} into
    \begin{equation}\label{eq:sylv_W}
        \resizebox{.875\columnwidth}{!}{$
        S_z^{}W^\top\big(F_B^\top(A-\lambda_0E)\big)=W^\top\big(F_B^\top E\big)+R_z\big(C+DL\widehat{\Pi}\big),
        $}
    \end{equation}
    which can be solved for $W$. A trivial approach (that often works in practical cases; see, for example, the one presented in Section~\ref{subsec:reduce}) would be to subsequently check whether, for the obtained $W$, the identity $\mathrm{rank}\big(W^\top F_B^\top(A-\lambda_0 E)\Pi\big)=\gamma$ holds and to verify that the linear equation system in \eqref{eq:sylv_W} is compatible, \emph{i.e.}, it admits a solution with respect to $G$.

    In contrast to the above-stated sequence of steps, it would be desirable (from not only a procedural perspective) to \emph{simultaneously compute} both $W$ and $G$, while at the same time freeing up the matrix $\widehat{\Pi}$ as a decision variable (which must implicitly be fixed, in order to first solve for $W$). The main challenge here is given by the inherent nonlinearity of both $W^\top F_B^\top(A-\lambda_0 E)\Pi G$ and $\mathrm{rank}\big(W^\top F_B^\top(A-\lambda_0 E)\Pi\big)$, with the latter term being a discontinuous function as well. Notably, the most readily available solution to combat these shortcomings is to employ the iterative optimisation scheme from \cite{BMI2LMI}, which has shown rapid convergence in practice compared to other techniques in literature (see, for example, the runtime analysis presented in Section~V-B of \cite{aug_sparse}).\newpage
    \begin{algorithm}\vspace{1mm}
		\textbf{Initialization:} Solve the constraints in \eqref{eq:prob_opt} while replacing $\Theta=O$ with $W-\widehat{W}=O$ for the tuple $\left(W^0,\widehat{W}^0,G^0,\widehat{\Pi}^0,P_{W}^0,P_{G}^0\right)$, use these matrices to form $T_A^0$, $T_B^0$, and $T_C^0$ as in \eqref{eq:aux_mat}, set $k\leftarrow0$, and then compute $f^0=\left\|T_C^{0}-T_A^{0}PT_B^{0}\right\|_*$\;
		
		\Repeat{$f^{k-1}-f^{k}<\eta_1$ or $f^k<\eta_2$}{
            $k\leftarrow k+1$\;\vspace{1mm}
            
			\eIf{$(k-1)\ \mathrm{mod}\ 2<1$}{
				Set $\Theta^k=P\left(T_B-T_B^{k-1}\right)$\;
			}{Set $\Theta^k=\left(T_A-T_A^{k-1}\right)P$\;}
			
			Solve $\mathcal{P}\left(-T_A^{k-1},-T_B^{k-1},\Theta^k\right)$ for the tuple $\left(W^k,\widehat{W}^k,G^k,\widehat{\Pi}^k,P_{W}^k,P_{G}^k\right)$ and use these matrices to form $T_A^{k}$, $T_B^{k}$, and $T_C^{k}$ as in \eqref{eq:aux_mat}\;\vspace{1mm}
			
			Compute $f^{k}=\left\|
			T_C^{k}-T_A^{k}PT_B^{k}\right\|_*$\;\vspace{2mm}
			
		}\vspace{2mm}\caption{Convex procedure for zero placement}\label{alg:iter}
	\end{algorithm}\vspace{-7mm}
    
    \noindent\hrulefill\vspace{0mm}
    
    As per Remark III.3 of \cite{aug_sparse}, the key lies in the fact that\vspace{-1mm}
    \begin{equation*}
        \mathrm{rank}\big(W^\top F_B^\top(A-\lambda_0 E)\Pi\big)=\gamma \iff W^\top \Psi W\succ O,\vspace{-1mm}
    \end{equation*}
    where $\Psi:= F_B^\top(A-\lambda_0 E)\Pi\Pi^\top(A-\lambda_0 E)^\top F_B$. This allows us to define the optimisation problem given in \eqref{eq:prob_opt}\stepcounter{equation}, located at the top of this page, having also defined the four matrices\vspace{-1mm}
    \begin{equation}\label{eq:aux_mat}
        \hspace{-1mm}\small\begin{array}{rcl}
                T_A&\hspace{-2mm}:=\hspace{-2mm}&\mathrm{diag}\left(W^\top,W^\top,O\right)\in\mathbb{R}^{(2\gamma+q_B)\times(2q_B+1)},\\
                P&\hspace{-2mm}:=\hspace{-2mm}&\mathrm{diag}\left(\Psi,F_B^\top(A-\lambda_0 E)\Pi,1\right)\in\mathbb{R}^{(2q_B+1)\times(q_B+\nu+1)},\\
                T_B&\hspace{-2mm}:=\hspace{-2mm}&\mathrm{diag}\left(\widehat{W},G,O\right)\in\mathbb{R}^{(q_B+\nu+1)\times(2\gamma+1)},\\
                T_C&\hspace{-2mm}:=\hspace{-2mm}&\mathrm{diag}\left(P_W,P_G,W-\widehat{W}\right)\in\mathbb{R}^{(2\gamma+q_B)\times(2\gamma+1)}.\vspace{-4.5mm}
            \end{array}\vspace{0.5mm}\normalsize
    \end{equation}
    The problem in \eqref{eq:prob_opt} will be solved as per the iterations in Algorithm~\ref{alg:iter}, shown above and inspired by Algorithm~1 in \cite{BMI2LMI}, for which we select two tolerance values $0<\eta_1,\eta_2\ll 1$. If the loop in Algorithm~\ref{alg:iter} terminates (the proposed procedure has guaranteed convergence; see \cite{BMI2LMI}) due to $f^k<\eta_2$, then the problem has been solved successfully. If the condition $f^{k-1}-f^k<\eta_1$ is the one that terminates the iteration, then the employed initialisation has likely converged to a non-zero local optimum, and another initialisation must be selected.\vspace{-2mm}

    \begin{remark}
        Note that condition \eqref{eq:not_pole} is not explicitly tackled at any point during the computational phase of our design procedure. This is owed to the fact that, statistically speaking, this condition will \emph{almost surely} be satisfied when imposing no restrictions upon the set of unconstrained eigenvalues belonging to $S-GL$ (as opposed to those fixed via the equalities in Proposition~\ref{prop:poles}). The proximity of these eigenvalues to the points $\mathfrak{f}(\mu_i)$ may yield a poorly conditioned solution to the moment-matching problem; however, this does not invalidate our theoretical results, and the optimisation of our solutions' conditioning is beyond the scope of this paper.\vspace{-2mm}
    \end{remark}

    We now demonstrate the efficiency of our technique on a practical problem, inspired by the electrical circuits in \cite{app}.

    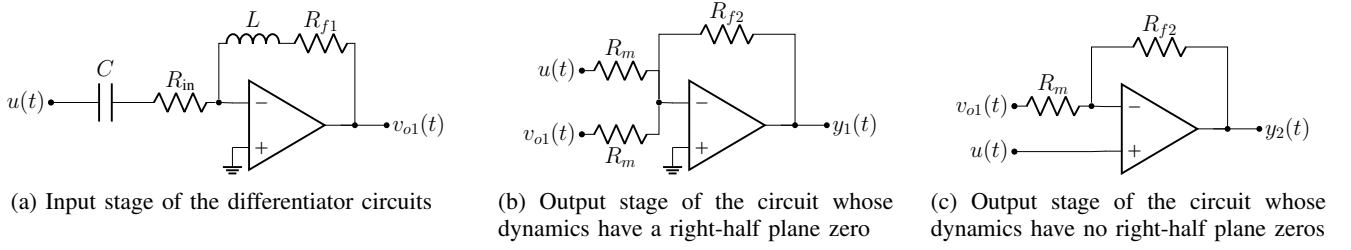
\begin{figure*}[h]
    \centering
    
    \subfloat[Input stage of the differentiator circuits\label{fig:input_stage}]{
        \scalebox{0.6}{
        \begin{tikzpicture}[line width=0.6pt]
            \draw (0,0) node[op amp] (opamp) {};
            
            \draw (opamp.+) -- (-1.2, -0.5) node[ground] {};
            
            \draw (opamp.-) -- (-1.5, 0.5) node[circle, fill, inner sep=1.3pt] (vminus) {};
            
            \draw (vminus) to[R, l_=\Large$R_{\text{in}}$] (-3.2, 0.5)
                           to[C, l_=\Large$C$] (-4.8, 0.5)
                           to[short, -*] (-5.2, 0.5) node[left] {\Large$u(t)$};
            
            \draw (vminus) -- (-1.5, 1.8)
                           to[L, l=\Large$L$] (0, 1.8)
                           to[R, l=\Large$R_{f1}$] (1.5, 1.8)
                           -- (1.5, 0) -- (opamp.out);
            
            \draw (opamp.out) to (2.2, 0) node[right] {\Large$v_{o1}(t)$};
            \draw (1.5, 0) node[circle, fill, inner sep=1.3pt]{};
            \draw (2.2, 0) node[circle, fill, inner sep=1.3pt]{};
        \end{tikzpicture}
        }
    }
    \hfill
    \subfloat[Output stage of the circuit whose dynamics have a right-half plane zero\label{fig:g1_unstable}]{
        \scalebox{0.6}{
        \begin{tikzpicture}[line width=0.6pt]
            \draw (0,0) node[op amp] (opamp) {};

            \draw (opamp.+) -- (-1.2, -0.5) node[ground] {};

            \draw (opamp.-) -- (-1.5, 0.5) node[circle, fill, inner sep=1.3pt] (vminus) {};

            \draw (vminus) -- (-1.5, 1.2)
                           to[R, l_=\Large$R_m$, -*] (-3.2, 1.2) node[left] {\Large$u(t)$};
            \draw (vminus) -- (-1.5, -0.2)
                           to[R, l=\Large$R_m$, -*] (-3.2, -0.2) node[left] {\Large$v_{o1}(t)$};

            \draw (vminus) -- (-1.5, 2.0)
                           to[R, l=\Large$R_{f2}$] (1.5, 2.0)
                           -- (1.5, 0) -- (opamp.out);

            \draw (opamp.out) to (2.2, 0) node[right] {\Large$y_1(t)$};
            \draw (1.5, 0) node[circle, fill, inner sep=1.3pt]{};
            \draw (2.2, 0) node[circle, fill, inner sep=1.3pt]{};
        \end{tikzpicture}
        }
    }
    \hfill
    \subfloat[Output stage of the circuit whose dynamics have no right-half plane zeros\label{fig:g2_stable}]{
        \scalebox{0.6}{
        \begin{tikzpicture}[line width=0.6pt]
            \draw (0,0) node[op amp] (opamp) {};
            
            \draw (opamp.+) -- (-1.5, -0.5)
                          to[short, -*] (-3.2, -0.5) node[left] {\Large$u(t)$};
            
            \draw (opamp.-) -- (-1.5, 0.5) node[circle, fill, inner sep=1.3pt] (vminus) {};
            \draw (vminus) -- (-1.5, 0.5)
                           to[R, l_=\Large$R_m$] (-3.2, 0.5)
                           to[short, -*] (-3.2, 0.5) node[left] {\Large$v_{o1}(t)$};
            
            \draw (vminus) -- (-1.5, 1.8)
                           to[R, l=\Large$R_{f2}$] (1.5, 1.8)
                           -- (1.5, 0) -- (opamp.out);
            
            \draw (opamp.out) to (2.2, 0) node[right] {\Large$y_2(t)$};
            \draw (1.5, 0) node[circle, fill, inner sep=1.3pt]{};
            \draw (2.2, 0) node[circle, fill, inner sep=1.3pt]{};
        \end{tikzpicture}
        }
    }
    
    \caption{Operational amplifier topologies that form the building blocks of the differentiation circuit from the numerical example}
    \label{fig:circuit_topologies}
    {\hrulefill}\vspace{-5mm}
\end{figure*}
    
	\section{Numerical Example}\label{sec:num_ex}

    \subsection{Modelling of the Descriptor System}\label{subsec:model}
        To highlight the capabilities of our proposed framework for computing a reduced-order model, while also preserving key structural elements of the original plant, we {hereby} aim to obtain a $6^{\text{th}}$-order continuous-time model that possesses a triple pole at $\{\infty\}$ and only one zero in $\{s\in\mathbb{C}\,\vert\,\mathrm{Re}(s) > 0\}$.\vspace{-1mm}

        \begin{remark}
            Since the system from our numerical example operates in continuous time, we adopt the specialised notation $\lambda=s$, and we subsequently refer to system-theoretical notions that pertain to the continuous-time context.\vspace{-1mm}
        \end{remark}

        Although the type of model we aim to obtain (recall Definition~\ref{def:inf_pz}, regarding infinite poles and zeros) is uncommon in practical applications, it can, however, be elegantly obtained by cascading three systems described by transfer functions of the following form
    \begin{align}\nonumber
    {\mathbf{G}_k(s) = \dfrac{b_{k2}s^2 + b_{k1}s + b_{k0}}{a_{k1}s + a_{k0}}, \, b_{kj} \neq 0,\,\forall\, j\in\{0:2\},}\\{a_{ki} \neq 0,\,\forall\,i\in\{0:1\},}\label{eq:desired_tf}
    \end{align}
    where each such transfer function has a pole at $\{\infty\}$. We now present the methods used to obtain these transfer functions.

    The transfer function from \eqref{eq:desired_tf} can be implemented using differentiator operational amplifier (op-amp) circuits, which are widely employed in practical applications, such as the circuit-based controllable capacitor shown in Figure~2 of~\cite{app}. This topology serves as the foundation for each of our three stages, providing the desired dynamics for this particular numerical example. To implement it, we add a resistance on the input path towards the negative terminal of the op-amp and an inductance on the feedback loop, as shown in Figure \ref{fig:input_stage} at the top of the page, while pointing out that this implementation is not typically used in practice. 

    Considering the dynamics of the op-amp displayed in Figure~\ref{fig:input_stage}, having a virtual ground at the inverting terminal which ensures that $V_+ = V_- = 0 \, \text{V}$, we can express the input voltage $u(t)$ on the input branch as\vspace{-1mm}
    \begin{equation}\label{eq:fig_volt}
        u(t) = v_C(t) + R_{\text{in}}i_{\text{in}}(t),\vspace{-1mm}
    \end{equation}
    where $i_{\text{in}}(t)$ is the current flowing through said branch. Knowing that $i_{\text{in}}(t)$ also flows through the capacitor and, therefore, can be expressed as \small$i_{\text{in}}(t) = C \diff{v_C(t)}{t}$\normalsize, differentiating the expression in \eqref{eq:fig_volt} with respect to time yields
    \begin{equation}\label{eq:modelling_1}
    \diff{u(t)}{t} = \left(\dfrac{1}{C} + R_{\text{in}}\diff{}{t}\right)i_{\text{in}}(t).
    \end{equation}
    
    By applying Kirchhoff's Voltage Law along the feedback branch between the inverting terminal node $V_-$ and the output node $v_{o1}(t)$, accounting for the fact that the output current is equal to the input current $i_{\text{in}}(t)$, we obtain\vspace{-1mm}
    \begin{equation}\label{eq:modelling_2}
    0 - v_{o1}(t) = \left(L \diff{}{t} + R_{f1}\right) i_{\text{in}}(t).
    \end{equation}\vspace{-2mm}

    \noindent
    We then apply the \scriptsize$\left(R_{\text{in}}C\diff{}{t} + 1\right)$ \normalsize operator to both sides of~\eqref{eq:modelling_2} and, since said operator commutes with \scriptsize$\left(L \diff{}{t} + R_{f1}\right)$\normalsize, we get\vspace{-1mm}
    \begin{equation*}
         \resizebox{\columnwidth}{!}{{$- \left(R_{\text{in}}C\diff{}{t} + 1\right)v_{o1}(t) = \left(L \diff{}{t} + R_{f1}\right) \left(R_{\text{in}}C\diff{}{t} + 1\right)i_{\text{in}}(t).$}}
    \end{equation*}\vspace{-3mm}

    \noindent
    Multiplying~\eqref{eq:modelling_1} by $C$ and then substituting it into the equality shown above yields the identity
    \[
    -\left(R_{\text{in}}C\diff{}{t} + 1\right)v_{o1}(t) = \left(L \diff{}{t} + R_{f1}\right)C\diff{u(t)}{t},
    \]\vspace{-3mm}

    \noindent
    which is equivalent to
    \[
    \left(R_{\text{in}}C\diff{}{t} + 1\right)v_{o1}(t) = -\left(LC \diff[2]{}{t} + R_{f1}C\diff{}{t}\right)u(t).
    \]\vspace{-3mm}

    \noindent
    Applying the Laplace transform under zero initial conditions yields the following transfer function from $u(t)$ to $v_{o1}(t)$
    \begin{equation}\label{eq:G_diff}
        \mathbf{G}_{\text{diff}}(s) = \dfrac{-LCs^2 - R_{f1}Cs}{R_{\text{in}}Cs + 1} =\dfrac{-s(LCs + R_{f1}C)}{R_{\text{in}}Cs + 1},
    \end{equation}\vspace{-3mm}

    \noindent
    which characterises the dynamics of the op-amp shown in Figure~\ref{fig:input_stage}. Although the obtained transfer function does have a pole at $\{\infty\}$, this also comes with a notable drawback.
    
    Comparing the expression in \eqref{eq:G_diff} with the one in~\eqref{eq:desired_tf}, it is immediately obvious that $\mathbf{G}_{\text{diff}}(s)$ has a zero {in $\{0\}$}. Thus, cascading three such systems yields a model with three zeros {in $\{0\}$}, which {does not align with our stated aims}. To avoid this fact, we design an additional output stage for our circuits, by using $u(t)$ and $v_{o1}(t)$ as inputs for a standard proportional op-amp circuit, shown in Figure~\ref{fig:g1_unstable} at the top of the page, with the resistor $R_{f2}$ placed in the feedback loop. 

    \begin{figure*}
        \begin{equation}\label{eq:g_cen}\tag{21}
            \textstyle
            \mathbf{G}(s) = \frac{(s-16)(s+0.4)(s+0.2)^2(s+4)^2}{(s+0.3)^3} =\left[ \tiny \begin{array}{rrrrrr|r}
                -0.5490 - s & 0.0163 & 0.8899 & -0.0477 & 0 & 0 & -116.7948\\
                -0.2666 & -0.0862-s & -0.0974 & 0 & 0 & 0 & -306.3973\\
                -0.0559 & 0.0318 & -0.2546-s & 0.1200 & 0 & 0 & -38.2123\\
                0 & 0 & 0 & -s & 4.0000 & 0 & -0.4000\\
                0 & 0 & 0 & 0 & -s & -4.0000 & 0.0200\\
                0.0238 & 0.0025 & -0.0751 & -0.9924 & 0 & 0 & 0.2481\\
                \hline
                -0.0281 & -0.0118 & 0.1436 & 1.2466 & 0.3038 & 0.1688 & 0
            \end{array}\right]_{-0.2}
            \normalsize.\hspace{-1mm}
        \end{equation}
        \hrulefill\vspace{-3mm}
    \end{figure*}

    Connecting both $u(t)$ and $v_{o1}(t)$ to the negative terminal of the op-amp, as in Figure~\ref{fig:g1_unstable}, yields the following input-output relation in the time domain\vspace{-1mm}
    \[
    y_1(t) = -\frac{R_{f2}}{R_m}(v_{o1}(t) + u(t)),\vspace{-1mm}
    \]
    which, after transitioning to the frequency domain as in~\eqref{eq:G_diff} and substituting $\mathbf V_{o1}(s)$ with $\mathbf G_{\text{diff}}(s)\mathbf U(s)$, produces
    \begin{equation*}
        \mathbf G_1(s) = \dfrac{R_{f2}LCs^2 + R_{f2}C(R_{f1} - R_{\text{in}})s - R_{f2}}{R_mR_{\text{in}}Cs + R_m},
    \end{equation*}
    which represents the transfer function of the circuit obtained by interconnecting the op-amps shown in Figures~\ref{fig:input_stage} and~\ref{fig:g1_unstable}.\newpage
    
    By employing now Vieta's formulas for the numerator, it is straightforward to check that $\mathbf{G}_1(s)$ is guaranteed to have a zero somewhere in $\{s\in\mathbb{C}\,\vert\,\mathrm{Re}(s) > 0\}$, with its counterpart located in $\{s\in\mathbb{C}\,\vert\,\mathrm{Re}(s) < 0\}$, for any positive values of the electrical components. Since we do not wish for any other zeros in the right half-plane or on the imaginary axis of the complex plane, we now proceed to design a similar output stage in which we connect $v_{o1}(t)$ to the negative terminal of the op-amp and $u(t)$ to the positive terminal, as depicted in Figure~\ref{fig:g2_stable} at the top of the previous page, thereby resulting in the following input-output relation\vspace{-1mm}
    \[
    {y_2(t) = -\frac{R_{f2}}{R_m}v_{o1}(t) + \left(1+\frac{R_{f2}}{R_m}\right)u(t).}\vspace{-1mm}
    \]
    In the frequency domain, this yields the transfer function
    \begin{equation*}
        {\mathbf{G}_2(s) \hspace{-0.5mm}=\hspace{-2mm} \large\begin{array}{l}
            \tfrac{R_{f2}LCs^2 + C(R_{f2}R_{f1} + R_{\text{in}}R_m + R_{\text{in}}R_{f2})s + (R_m + R_{f2})}{R_mR_{\text{in}}Cs + R_m}
        \end{array}\hspace{-1.5mm},\normalsize}
    \end{equation*}
    which represents the transfer function of the circuit obtained by interconnecting the op-amps shown in Figures~\ref{fig:input_stage} and~\ref{fig:g2_stable}, and for which all zeros are guaranteed to be located in the open left-half complex plane, as per the Hurwitz criterion.

    We now choose a set of convenient (although not typically encountered in practice) values for the electrical components of these circuits, in order to obtain a model with a desirable finite pole-zero structure. Hence, for $\mathbf{G}_1(s)$, we consider the values $R_{\text{in}} = \frac{40}{3} \, \Omega$, $R_{f1} = \frac{43}{12} \, \Omega$, $R_{f2} = \frac{1}{10} \, \Omega$, $R_{m} = 1 \, \Omega$, $C = 250 \, \text{mF}$, and $L = \frac{5}{8} \, \text{H}$. For $\mathbf{G}_2(s)$, we proceed to take $R_{\text{in}} = \frac{10}{3} \, \Omega$, $R_{f1} = \frac{23}{6} \, \Omega$, $R_{f2} = 1 \, \Omega$, $R_{m} = 1 \, \Omega$, $C = 1000 \, \text{mF}$, and $L = \frac{5}{2} \, \text{H}$. Therefore, the final model to be used in the sequel will be given by
    \begin{equation*}
        \mathbf{G}(s) = \mathbf{G}_1(s)\mathbf{G}_2(s)\mathbf{G}_2(s),
    \end{equation*}
    which, for the physical values specified above, is described by the transfer function in \eqref{eq:g_cen}\stepcounter{equation}, at the top of this page, and by the minimal realisation centred in $\lambda_0=-0.2$ given therein.

    \subsection{Moment Matching with Pole-Zero Constraints}\label{subsec:reduce}
	
	We are interested in obtaining a reduced-order model that preserves all the poles and zeros of $\mathbf{G}(s)$ which are located in the set $\{s\in\mathbb{C}\,\vert\,\mathrm{Re}(s)\geq 0\}\cup\{\infty\}$, as knowledge pertaining to these structural elements is essential for the stabilisation and control of a plant in closed-loop configuration (see, for example, Chapter~5 in \cite{zhou}). Since $\mathbf{G}(s)$ has a triple pole at $\{\infty\}$ along with a zero in $s_z=16$, we will set $\nu=4$ and will preserve all of the aforementioned structural elements in the reduced-order model. Moreover, given the frequency characteristic of $\mathbf{G}(s)$ as a differentiator, in which the high-frequency gain grows unboundedly due to the poles at $\{\infty\}$, practical considerations motivate us to preserve the filter's profile at low frequencies. Thus, we proceed to match the $0^\text{th}$-order moments of $\mathbf{G}(s)$ at $\mu_i\in\{-0.1j,0.1j,-j,j\}$, and we select $p_k\in\{\lambda_0+5\cdot10^4,\lambda_0+10^5,\infty\}$ along with $z_\ell=16$.\vspace{-1.5mm}
	
	\begin{remark}\label{rem:approx}
		Although the poles preserved in the reduced-order model must be distinct (recall Section~\ref{subsec:prob_st}), notice that $\mathfrak{f}(p_k)\in\{0,10^{-5},2\cdot10^{-5}\}$, thus enabling a round-off approximation to a triple pole in $s=0$ for $\mathbf{H}(\mathfrak{f}(s))$. The effects of this approximation will be addressed in the sequel.\vspace{-1.5mm}
	\end{remark}
	
	We therefore take $L=\begin{bmatrix}
		1 & 1 & 1 & 1
	\end{bmatrix}$ along with \vspace{-1.5mm}
	\begin{equation*}
		S=\mathrm{diag}\left(\begin{bmatrix} \phantom{-}0.1923 & 0.9615 \\ -0.9615 & 0.1923 \end{bmatrix}, \begin{bmatrix} \phantom{-}4 & 2 \\ -2 & 4 \end{bmatrix}\right),\vspace{-1.5mm}
	\end{equation*}
	in Theorem~\ref{thm:class}, so as to obtain the solution of \eqref{eq:gen_sylv}, namely\vspace{-1mm}
    \begin{equation*}
        \Pi=\left[\tiny\begin{array}{rrrrrr}
            116.0611 & 305.9632  & 37.2368  &  0.4276  &  0.0574  &  0.0007\\
            114.5322 & 306.7256  & 39.6616  &  0.7095  & -0.0269  &  0.0183\\
            118.2258 & 309.8884  & 41.3876  &  0.4002  & -0.0200  & -0.0010\\
            119.8748 & 296.4812  & 42.1431  &  0.3989  & -0.0200  & -0.0010
        \end{array}\right]^\top\hspace{-0.5mm},\vspace{-1mm}
    \end{equation*}
    while setting $S_z=\mathfrak{f}(16)$ alongside $R_z=1$ and choosing $\widehat{\Pi}$ as the Moore-Penrose pseudoinverse of $\Pi$ in Theorem~\ref{thm:zeros}. These values generate the gain matrix\vspace{-1.5mm}
	\begin{equation*}
		G=\begin{bmatrix}
			-0.1317 & -0.1344 & -1.5149 & 15.0497
		\end{bmatrix}^\top,\vspace{-1.5mm}
	\end{equation*}
	which produces the reduced-order model\vspace{-1mm}
	\begin{equation*}
		\mathbf{H}(\mathfrak{f}(s)) =\hspace{-1mm} \scriptsize\begin{array}{l}
	       \dfrac{ -3.802 s^3 - 0.9435 s^2 - 0.5925 s + 0.0411}{ s^4 + 4.884 s^3 - 1.3\cdot10^{-4} s^2 + 2\cdot 10^{-10} s - 2\cdot 10^{-15}}\ ,
		\end{array}\normalsize\vspace{-1mm}
	\end{equation*}
	that we approximate (recall Remark~\ref{rem:approx}) into\vspace{-1mm}
	\begin{equation}\label{eq:approx_init}
		\hspace{-2mm}\widehat{\mathbf{H}}(\mathfrak{f}(s)) =
		\frac{ -3.802 s^3 - 0.9435 s^2 - 0.5925 s + 0.0411}{ s^4 + 4.884 s^3},\vspace{-1mm}
	\end{equation}
	yielding the fourth-order approximant of $\mathbf{G}(s)$ given by\vspace{-1mm}
	\begin{equation}\label{eq:approx_final}
		\hspace{-1mm}\widehat{\mathbf{H}}(s) =\hspace{-1mm} \small\begin{array}{l}
			\dfrac{ 0.0084(s-16)(s+0.2)(s^2+2.171s+6.109)}{s + 0.4047}
		\end{array}\hspace{-1mm}.\normalsize\vspace{-1mm}
	\end{equation}
	Not only does the transfer function displayed in \eqref{eq:approx_final} have the desired pole-zero structure, but the approximation undertaken in \eqref{eq:approx_init} only marginally degrades the performed moment matching at $\pm0.1j$ and $\pm j$, and it is simple to check that\vspace{-1mm}
	\begin{equation*}
        \resizebox{\columnwidth}{!}{$
		\begin{array}{rcl}
			\left|\Big|\eta_{\,\mathbf{G},0}(\pm0.1j)\Big|-\left|\eta_{\,\widehat{\mathbf{H}},0}(\pm0.1j)\right|\right|\hspace{-2.5mm}&\leq&\hspace{-2.5mm}3\cdot 10^{-6}\cdot\left|\eta_{\,\mathbf{G},0}(\pm0.1j)\right|,\\\vspace{-3mm}\\
			\left|\Big|\eta_{\,\mathbf{G},0}(\pm j)\Big|-\left|\eta_{\,\widehat{\mathbf{H}},0}(\pm j)\right|\right|\hspace{-2.5mm}&\leq&\hspace{-2.5mm}3\cdot 10^{-7}\cdot\left|\eta_{\,\mathbf{G},0}(\pm j)\right|,\\\vspace{-3mm}\\
			\left|\mathrm{arg}(\eta_{\,\mathbf{G},0}(\pm0.1j))-\mathrm{arg}(\eta_{\,\widehat{\mathbf{H}},0}(\pm0.1j))\right|\hspace{-2.5mm}&\leq&\hspace{-2.5mm}8\cdot 10^{-7}\cdot\left|\mathrm{arg}(\eta_{\,\mathbf{G},0}(\pm0.1j))\right|,\\\vspace{-3mm}\\
			\left|\mathrm{arg}(\eta_{\,\mathbf{G},0}(\pm j))-\mathrm{arg}(\eta_{\,\widehat{\mathbf{H}},0}(\pm j))\right|\hspace{-2.5mm}&\leq&\hspace{-2.5mm}2\cdot 10^{-5}\cdot\left|\mathrm{arg}(\eta_{\,\mathbf{G},0}(\pm j))\right|,
		\end{array}$}\vspace{-1mm}
	\end{equation*}
	thus confirming the practical validity of said approximation.\vspace{-1mm}
	
	\section{Conclusion}\label{sec:outro}

    By using M\"obius mappings in conjunction with centred realisations, the problem of model order reduction with moment matching for descriptor systems can be elegantly solved by adapting standard, state-space-based theory. Additionally, due to the properties of said reversible mappings, imposing pole-zero constraints on the reduced-order model amounts to a mere remapping of the (extended) complex plane. Finally, given that this approach has shown promising results in the single-input single-output case, future research efforts on this topic can focus on the multiple-input multiple-output case.\vspace{-1mm}

    \begin{figure*}
    \begin{equation}\label{eq:cei_faa_asta}\tag{29}
            \small\begin{array}{c}
                \mathbf{G}^{(j)}(\mu_i)=\displaystyle\sum_{q=1}^j \mathbf{\widetilde{G}}^{(q)}(\mathfrak{f}(\mu_i)) \cdot \mathcal{B}_{j,q}\left(\mathfrak{f}'(\mu_i), \dots, \mathfrak{f}^{(j-q+1)}(\mu_i)\right),\ \mathbf{H}^{(j)}(\mu_i)=\displaystyle\sum_{q=1}^j \mathbf{\widehat{G}}^{(q)}(\mathfrak{f}(\mu_i)) \cdot \mathcal{B}_{j,q}\left(\mathfrak{f}'(\mu_i), \dots, \mathfrak{f}^{(j-q+1)}(\mu_i)\right).
            \end{array}\normalsize\vspace{-1mm}\hspace{-2mm}
        \end{equation}
        \hrulefill\vspace{-6mm}
    \end{figure*}

	\section*{Appendix}
		
		\noindent\textbf{Proof of Theorem~\ref{thm:class}
		} 
		
		\noindent
		To prove point a), we first show that the Sylvester equation\vspace{-1.5mm}
		\begin{equation}\label{eq:sylv_aux}
			(A-\lambda_0 E)^{-1}E\Pi+(A-\lambda_0 E)^{-1}BL=\Pi S\vspace{-1.5mm}
		\end{equation}
		has a unique solution, which is equivalent to $(A-\lambda_0 E)^{-1}E$ and $S$  not having any common eigenvalues. Since $S$ is similar to the matrix $J$, given in the statement, the eigenvalues of $S$ are precisely $\mathfrak{f}(\mu_i)$. Therefore, we need only check that\vspace{-1.5mm}
		\begin{equation}\label{eq:no_comm}
			\det\left(\mathfrak{f}(\mu_i)(A-\lambda_0 E)-E\right)\neq 0,\,\forall\,i\in\{1:\alpha\}.\vspace{-1.5mm}
		\end{equation}
		We select a generic $i\in\{1:\alpha\}$ and we begin by treating the case $\mathfrak{f}(\mu_i)=0$, which is equivalent to $\{\mu_i\}=\{\infty\}$. Since $\mu_i\not\in\Lambda(A-\lambda E)$ (recall the Problem Statement given in Section~\ref{subsec:prob_st}), we get the fact that $\det (E)\neq0$, since the pencil $A-\lambda E$ can have no infinite generalised eigenvalues, from which \eqref{eq:no_comm} follows. If we have $\mathfrak{f}(\mu_i)\neq0$, then we have $
		\det(\mathfrak{f}(\mu_i)(A-\lambda_0 E)-E)=\frac{1}{(\mu_i-\lambda_0)^n}\det(A-\mu_i E)$. Since now $\mu_i\in\mathbb{C}\setminus\{\lambda_0\}$ along with $\mu_i\not\in\Lambda(A-\lambda E)$ (recall, once again, Section~\ref{subsec:prob_st}), then $\Pi$ is the unique solution of \eqref{eq:sylv_aux}. Moreover, due to the nonsingularity of the matrix $A-\lambda_0 E$, it follows that, by left-multiplying with $A-\lambda_0 E$ in \eqref{eq:sylv_aux} and with its inverse in \eqref{eq:gen_sylv}, $\Pi$ is a solution of \eqref{eq:sylv_aux} if and only if it solves \eqref{eq:gen_sylv}. Consequently, its uniqueness as a solution of \eqref{eq:gen_sylv} follows from it being the unique solution of \eqref{eq:sylv_aux}.
		
		The proof of point b) is more involved, and can be broken down into 4 parts: \textbf{I)} we first show that there exists a transfer function $\widehat{\mathbf{G}}(\lambda)$ that matches the moments of $\widetilde{\mathbf{G}}(\lambda)$ (recall point b) of Lemma~\ref{lem:remap}) at each $\mathfrak{f}(\mu_i)$, \textbf{II)} we then prove that $\mathbf{H}(\lambda)$ from \eqref{eq:has} is none other than $\widehat{\mathbf{G}}(\mathfrak{f}(\lambda))$, \textbf{III)} we establish that its moments are well-defined at each point $\mu_i$, and \textbf{IV)} we explain how objective A) is a consequence of \textbf{I)} and \textbf{II)}.
		
		\textbf{I)} Denoting $\widetilde{A}:=(A-\lambda_0E)^{-1}E$, $\widetilde{B}:=(A-\lambda_0E)^{-1}B$, and recalling the Sylvester equation given in \eqref{eq:sylv_aux}, it follows from equations (4) and (5) in \cite{PZ_MM} that the state-space system $(S-GL,G,C\Pi,O)$, which realises the transfer function\vspace{-2mm}
		\begin{equation}\label{eq:G_hat}
			\widehat{\mathbf{G}}(\lambda)-D = C\Pi(\lambda I_{\nu}-S+GL)^{-1}G,\vspace{-2mm}
		\end{equation}
		achieves moment matching for $\widetilde{\mathbf{G}}(\lambda)-D=C(\lambda I_\nu-\widetilde{A})^{-1}\widetilde{B}$ as per the identities $\eta_{\,\widetilde{\mathbf{G}}-D,j}(\mathfrak{f}({\mu_i})) = \eta_{\,\widehat{\mathbf{G}}-D,j}(\mathfrak{f}({\mu_i}))$, for all $i \in \{1:\alpha \}$ and all $j \in \{0:n_i-1\}$,
		so long as \eqref{eq:not_pole} holds.
		
		Recalling Definition~\ref{def:mom} and adding back the $D$ matrix to both of the aforementioned transfer functions, it follows that\vspace{-2mm}
		\begin{equation}\label{eq:match_tran}
			\resizebox{.9\columnwidth}{!}{$
            \eta_{\,\widetilde{\mathbf{G}},j}(\mathfrak{f}({\mu_i})) = \eta_{\,\widehat{\mathbf{G}},j}(\mathfrak{f}({\mu_i})),\,\forall\, i \in \{1:\alpha \},\, j \in \{0:n_i-1\},
            $}\hspace{-1mm}\vspace{-1mm}
		\end{equation}
		by direct differentiation with respect to a sum of terms, since the $D$ matrix is a constant with respect to the $\lambda$ variable.
		
		\textbf{II)} Let $\mathbf{H}(\lambda):=\widehat{\mathbf{G}}(\mathfrak{f}(\lambda))$ and effect the change of variable $\lambda\mapsto\mathfrak{f}(\lambda)$ in \eqref{eq:G_hat} to get that\vspace{-3mm}
		\begin{equation*}
			\small\begin{array}{rcl}
			     \mathbf{H}(\lambda) &\hspace{-2mm}=\hspace{-2mm}& C \Pi \left( \tfrac{1}{\lambda - \lambda_0} I_{\nu} - S + GL \right)^{-1}G+D\\
			   &\hspace{-2mm}=\hspace{-2mm}&C\Pi \left(-I_{\nu} + (\lambda_0 - \lambda)(-S+GL)\right)^{-1}G(\lambda_0 - \lambda)+D.\vspace{-1mm}
			\end{array}
		\end{equation*}
		\noindent Rearranging terms inside the bracket that is being inverted, one directly retrieves \eqref{eq:has}. Note also that the pole pencil of $\mathbf{H}(\lambda)$ given in \eqref{eq:has} satisfies $\det(I_{\nu}+(\lambda_0-\lambda)(S-GL))\not\equiv0$, since $\det(I_{\nu}+(\lambda_0-\lambda)(S-GL))\big\vert_{\lambda=\lambda_0}=1$, which makes the transfer function $\mathbf{H}(\lambda)$ well-defined for all $G$ in \eqref{eq:has}.
		
		\textbf{III)} In order to show that $\eta_{\,\mathbf{H},j}(\mu_i)$ is well-defined for all $i\in\{1:\alpha\}$ and all $j\in\{0:n_i-1\}$, it suffices to prove that $\mu_i\not\in\Lambda(I_{\nu} + \lambda_0(S-GL)-\lambda(S-GL))$ for each $i\in\{1:\alpha\}$. This would imply that no $\mu_i$ is a pole of $\mathbf{H}(\lambda)$, thus making said rational function analytical at each $\mu_i$ and, consequently, making the derivatives that go into $\eta_{\,\mathbf{H},j}(\mu_i)$ be well-defined.
		
		To this end, select any $i\in\{1:\alpha\}$, recall \eqref{eq:not_pole} and consider first the case $\mathfrak{f}(\mu_i)=0$, which is equivalent to $\{\mu_i\}=\{\infty\}$. By \eqref{eq:not_pole}, it follows that $\det(S-GL)\neq 0$, which implies that the pole pencil from \eqref{eq:has} cannot have infinite generalised eigenvalues. If we have $\mathfrak{f}(\mu_i)\neq 0$, then \eqref{eq:not_pole} ensures that\vspace{-2mm}
		\begin{multline}\label{eq:not_eig}
			\tfrac{1}{(\mu_i-\lambda_0)^\nu}\det(I_\nu-(\mu_i-\lambda_0)(S-GL))=\\=\det\left(\tfrac{1}{\mu_i-\lambda_0}I_\nu-S+GL\right)\neq 0.
		\end{multline}
		\vspace{-5mm}
		
		\noindent Since now $\mu_i\in\mathbb{C}\setminus\{\lambda_0\}$, we multiply \eqref{eq:not_eig} by $(\mu_i-\lambda_0)^\nu$ to once again get that $\mu_i\not\in\Lambda(I_{\nu} + \lambda_0(S-GL)-\lambda(S-GL))$.
		
		\textbf{IV)} The fact that objective A) holds for $j=0$ and for all $i\in\{1:\alpha\}$ follows directly from \eqref{eq:match_tran}, since\vspace{-1mm}
		\begin{equation*}
			\resizebox{\columnwidth}{!}{
            $\eta_{\,\mathbf{G},0}(\mu_i)=\mathbf{G}(\mu_i)=\widetilde{\mathbf{G}}(\mathfrak{f}(\mu_i))=\widehat{\mathbf{G}}(\mathfrak{f}(\mu_i))=\mathbf{H}(\mu_i)=\eta_{\,\mathbf{H},0}(\mu_i),$
            }
			\vspace{-1mm}
		\end{equation*}
		for all $i\in\{1:\alpha\}$. For those cases where $j>0$, it suffices to show that $\mathbf{G}^{(j)}(\mu_i)=\mathbf{H}^{(j)}(\mu_i),\,\forall\,i\in\{1:\alpha\}$. Recalling that $\mathbf{G}(\lambda)=\widetilde{\mathbf{G}}(\mathfrak{f}(\lambda))$ and $\mathbf{H}(\lambda)=\widehat{\mathbf{G}}(\mathfrak{f}(\lambda))$, we select an arbitrary $i\in\{1:\alpha\}$ and we employ \emph{Fa\`a di Bruno's Formula} for higher-order differentiation to obtain the identities given in \eqref{eq:cei_faa_asta}\stepcounter{equation}, which is located at the top of this page, where $\mathcal{B}_{j, q}(\cdot)$ is the \emph{partial Bell polynomial} of order $j$ and degree $q$ (see Section~2 of \cite{johnson_fa_di_bruno} for an overview of Bell polynomials and Fa\`a di Bruno's Formula). Recall that \eqref{eq:match_tran} holds, implying that $\mathbf{\widetilde{G}}^{(q)}(\mathfrak{f}(\mu_i))=\mathbf{\widehat{G}}^{(q)}(\mathfrak{f}(\mu_i))$ for all $i\in\{1:\alpha\}$ and all $q\in\{0:n_i-1\}$. Since $j\leq n_i-1$ for all $i\in\{1:\alpha\}$, we get that all the terms of the two sums in \eqref{eq:cei_faa_asta} are pairwise equal. Thus, so are $\mathbf{G}^{(j)}(\mu_i)$ and $\mathbf{H}^{(j)}(\mu_i)$, for all $i\in\{1:\alpha\}$ and $j\in\{0:n_i-1\}$, from which the desired result follows.\qed
		
	\smallskip
	
	\noindent
	\textbf{Proof of Proposition~\ref{prop:poles}}
	
	\noindent
	To prove this statement, we begin by defining the matrices\vspace{-1mm}
	\begin{equation*}
		M_k := \scriptsize\begin{bmatrix}
			1 & L\\
			G & -D_k
		\end{bmatrix}, \,\forall\,k\in\{1:\beta\}.\vspace{-1mm}\normalsize
	\end{equation*}
	Since all $p_k$ and $\mu_i$ are pairwise distinct (recall Section~\ref{subsec:prob_st}) and since $\mathfrak{f}(\cdot)$ is invertible (recall Lemma~\ref{lem:inv}), making it bijective, it follows that $\mathfrak{f}(p_k)$ and $\mathfrak{f}(\mu_i)$ are also pairwise distinct. Thus, the square matrices $D_k$ defined in the result's statement have no zero eigenvalues, making them all invertible. In light of this fact, we now apply the \emph{Schur determinant formula} (see, for example, Section~0.8.5 of \cite{matrix_analysis}) to get that \vspace{-2mm}
	\begin{multline}\label{eq:det}
		\det M_k = \det 1 \cdot \det \left(1-L(-D_k)^{-1}G\right)= \\= \det (-D_k) \cdot \det \left(-D_k - G (1)^{-1}L\right),\vspace{-1mm}
	\end{multline}
	\vspace{-5mm}
	
	\noindent for all $k\in\{1:\beta\}$. By recalling that $D_k^{-1}=\mathcal R_k+j\mathcal I_k$ (with $j:=\sqrt{-1}$) and by multiplying the identities from \eqref{eq:pole_cond} to the left with $\begin{bmatrix}
		1 & j
	\end{bmatrix}$, we get that $\det (1-L(-D_k)^{-1}G) = 0$, which implies that $\det M_k = 0,$ for all $k\in\{1:\beta\}$, and consequently, the final term from \eqref{eq:det} must also be zero. Since $D_k$ is invertible ($\det (-D_k) \neq 0$), it follows that \vspace{-1mm}
	\begin{equation*}
		\det \left(-D_k - G(1)^{-1}L\right) = \det \left(S-GL-\mathfrak{f}(p_k) I_\nu\right) = 0.
	\end{equation*}
	\vspace{-6mm}
	
	\noindent We must now consider two cases. If $p_k\in\mathbb{C}\setminus\{\lambda_0\}$, then\vspace{-2mm}
	\begin{multline}\label{eq:fin_pole}
		0 = \det \left(\tfrac{1}{p_k-\lambda_0} I_\nu-(S-GL)\right)=\\=\tfrac{1}{(p_k-\lambda_0)^\nu}\det(I_\nu-(p_k-\lambda_0)(S-GL)),
	\end{multline}
	\vspace{-5mm}
	
	\noindent and, since $\frac{1}{(p_k-\lambda_0)}\in\mathbb{C}\setminus\{0\}$, the determinant from the rightmost term in \eqref{eq:fin_pole} must be $0$. Recalling that the determinant of the pole pencil from \eqref{eq:has} is not identically 0, the desired conclusion follows. If $\{p_k\}=\{\infty\}$, then we get $\mathfrak{f}(p_k)=0$ and $\det(S-GL)=0$. Once again, due to the determinant of the pole pencil from \eqref{eq:has} not being identically 0, said pencil must have an infinite generalised eigenvalue. \qed\smallskip
	
	\noindent
	\textbf{Proof of Theorem~\ref{thm:zeros}} 
	
	\noindent
	To prove point a), we first show that $\widehat{\Pi}$ does indeed exist. Recall from Section~\ref{subsec:prob_st} that $\delta(G)=n$ and, therefore, by Proposition~\ref{prop:min_css} and point c) of Lemma~\ref{lem:remap}, the state-space system $\big(\widetilde{A},\widetilde{B},C,D\big)$ is minimal, where $\widetilde{A}:=(A-\lambda_0E)^{-1}E$ along with $\widetilde{B}:=(A-\lambda_0E)^{-1}B$. Combining this minimality with the fact that the pair $(L, S)$ is observable and that $\Pi$ is also the solution of \eqref{eq:sylv_aux}, it follows (see Section~2 of \cite{PZ_MM}) that $\Pi$ has full column rank and a left-inverse denoted $\widehat{\Pi}$.
	
	Moving on, we show that the Sylvester equation, highlighted in \eqref{eq:sylv_z} and written in equivalent form as\vspace{-1.5mm}
	\begin{equation}\label{eq:sylv_alt}
		S_z\widetilde{\Upsilon}_z=\widetilde{\Upsilon}_z \widetilde{A}+R_z\big(C+DL\widehat{\Pi}\big),\vspace{-1.5mm}
	\end{equation}
    having defined $\widetilde{\Upsilon}_z:=\Upsilon_z(A-\lambda_0E)$, has a unique solution. This fact is equivalent to $\widetilde{A}$ and $S_z$ not having any common eigenvalues, with the subsequent proof being analogous, \emph{mutatis mutandis}, to the one given for point a) of Theorem~\ref{thm:class}. Therefore, we hereby omit it for the sake of brevity.
	
		To prove point b), we first proceed to retrieve the fact that\vspace{-1mm}
		\begin{equation}\label{eq:vpg_1}
			\det \scriptsize\begin{bmatrix}
				S-GL- \mathfrak{f}(z_\ell)I_{\nu}& G\\
				C\Pi & D
			\end{bmatrix} = 0, \,\forall\,\ell \in\{1:\gamma\},\normalsize\vspace{-1mm}
		\end{equation}
        and we treat the two cases from \eqref{eq:zero_cond} separately, to show that objective C) from Section~\ref{subsec:prob_st} is satisfied in both instances. To begin, we point out that $S_z\widetilde{\Upsilon}_z=\widetilde{\Upsilon}_z\widetilde{A} +R_z\big(C+DL\widehat{\Pi}\big)$ has a unique solution, with $\mathrm{rank}\big(\widetilde{\Upsilon}_z\Pi\big) = \gamma$ and $\widetilde{\Upsilon}_z\widetilde B=O$, due to the assumptions from the statement of point b).
	
	{
		Since $\mathfrak{f}(z_\ell) \in \Lambda(S_z-\lambda I_\gamma)$, there exists $w \in \mathbb{C}^\gamma, \, w \neq O$,\newline such that $w^\top(S_z-\mathfrak{f}(z_\ell)I_\gamma ) = O$. Right-multiplying this identity with $\widetilde{\Upsilon}_z\Pi$ yields $w^\top(S_z\widetilde{\Upsilon}_z\Pi-\mathfrak{f}(z_\ell)\widetilde{\Upsilon}_z\Pi) = O$, and substituting $S_z\widetilde{\Upsilon}_z$ with the expression in~\eqref{eq:sylv_alt}, we obtain \vspace{-1.5mm}
		\begin{equation}\label{eq:helper_dem_z}
			\hspace{-1mm}w^\top\big( \widetilde{\Upsilon}_z\widetilde{A}\Pi + R_z \big(C+DL\widehat{\Pi}\big)\Pi - \mathfrak{f}(z_\ell)\widetilde{\Upsilon}_z\Pi\big) = O.\vspace{-1.5mm}
		\end{equation}
        Recalling the fact that $\widetilde{\Upsilon}_z := \Upsilon_z(A-\lambda_0E)$ and then left-multiplying in~\eqref{eq:gen_sylv} with $\widetilde{\Upsilon}_z(A-\lambda_0E)^{-1}$ results in \vspace{-1.5mm}
		\begin{equation*}
			\widetilde{\Upsilon}_z\widetilde{A} \Pi + \widetilde{\Upsilon}_z\widetilde BL = \widetilde{\Upsilon}_z\Pi S.\vspace{-1.5mm}
		\end{equation*}
		Substituting this identity into~\eqref{eq:helper_dem_z} and employing the fact that $\widetilde{\Upsilon}_z\widetilde B = O$, as indicated in the result's statement, yields\vspace{-1mm}
		\begin{equation*}
			w^\top\big(\widetilde{\Upsilon}_z\Pi(S - \mathfrak{f}(z_\ell)I_\nu) + R_z \big(C+DL\widehat{\Pi}\big)\Pi\big) = O,\vspace{-1mm}
		\end{equation*}
        which, in matrix form, is equivalent to \vspace{-1.5mm}
		\begin{equation*}
			\scriptsize\begin{bmatrix}
				w^\top \widetilde{\Upsilon}_z\Pi & w^\top R_z
			\end{bmatrix}
			\scriptsize\begin{bmatrix}
				S - \mathfrak{f}(z_\ell)I_\nu \\ C\Pi+DL
			\end{bmatrix} = O,\vspace{-1.5mm}
		\end{equation*}
		and which we further extend into the following identity \vspace{-1mm}
		\begin{equation}\label{eq:left_null}
			\hspace{-1mm}\resizebox{.9\columnwidth}{!}{
            $\begin{bmatrix}
				w^\top \widetilde{\Upsilon}_z\Pi & w^\top R_z
			\end{bmatrix}
			\begin{bmatrix}
				S - \mathfrak{f}(z_\ell)I_\nu & G \\ C\Pi+DL & D
			\end{bmatrix} = \begin{bmatrix}O &  \widetilde{\Upsilon}_z\Pi G + R_zD \end{bmatrix}.$
            }
		\end{equation}
		\vspace{-5mm}
		
		Recall now, for the result's statement, that $\widetilde{\Upsilon}_z\Pi$ has full row rank and that $w\neq O$, from which it follows that both $w^\top\widetilde{\Upsilon}_z\Pi$ and the row vector from the left-hand side of \eqref{eq:left_null} are non-zero. Moreover, since we have $\widetilde{\Upsilon}_z\Pi G + R_zD = O$ holds due to \eqref{eq:suff_sing}, we must also have that the matrix from the left-hand side of \eqref{eq:left_null} is singular, yielding the fact that\vspace{-1mm}
		\begin{equation*}
			\resizebox{\columnwidth}{!}{
            $\det \begin{bmatrix}
				S -GL - \mathfrak{f}(z_\ell)I_\nu & G \\ C\Pi & D
			\end{bmatrix}=\det\begin{bmatrix}
				S - \mathfrak{f}(z_\ell)I_\nu & G \\ C\Pi+DL & D
			\end{bmatrix}\det\begin{bmatrix}
				I_\nu & O\\-L & 1
			\end{bmatrix}=0,$
            }\vspace{-1mm}
		\end{equation*}
        from which we retrieve~\eqref{eq:vpg_1}. We now show that the implications in \eqref{eq:zero_cond} hold, starting with the one on the left-hand side.

	{ 
		Define $Z_\ell:=\mathrm{diag}((\lambda_0-z_\ell)I_\nu,1)$ and employ \eqref{eq:vpg_1} to get \vspace{-5mm}
		\begin{multline}\label{eq:vpg_2}
			\det \scriptsize\begin{bmatrix}
				I_\nu + (\lambda_0-z_\ell)(S-GL)& G(\lambda_0 - z_\ell)\\
				C\Pi & D
			\end{bmatrix} =\\=\det Z_\ell\det \scriptsize\begin{bmatrix}
				S-GL- \mathfrak{f}(z_\ell)I_{\nu}& G\\
				C\Pi & D
			\end{bmatrix}= 0,\normalsize
		\end{multline}
		\vspace{-5mm}
		
		\noindent for all $\ell \in\{1:\gamma\}$ and $z_\ell\in\mathbb{C}\setminus\{\lambda_0\}$, from the determinant's distributivity over matrix multiplication and, since $\mathbf{H}(\lambda)\not\equiv0$ (recall Remark~\ref{rem:not_id_0}), we have (see Chapter~3 in \cite{zhou}) that\vspace{-1mm}
		\begin{equation}\label{eq:vpg_3}
			\det \scriptsize\begin{bmatrix}
				I_\nu + (\lambda_0-\lambda)(S-GL)& G(\lambda_0 - \lambda)\\
				C\Pi & D
			\end{bmatrix} \not\equiv 0.\normalsize
		\end{equation} 
		\vspace{-4mm}
		
		\noindent By combining \eqref{eq:vpg_2} and \eqref{eq:vpg_3}, the left-hand side of \eqref{eq:zero_cond} follows.
		
		{
			Finally, to retrieve the implication from the right-hand side of \eqref{eq:zero_cond}, we first show that if $\{z_\ell\}=\{\infty\}$, then $\begin{bmatrix}
		          S-GL & G
	        \end{bmatrix}$ must have full row rank. Recall now from Section~\ref{subsec:prob_st} that all $z_\ell$ and $\mu_i$ are pairwise distinct. Thus, due to $\mathfrak{f}(\cdot)$ being invertible (recall Lemma~\ref{lem:inv}), making it bijective, we get that all $\mathfrak{f}(z_\ell)$ and $\mathfrak{f}(\mu_i)$ are also pairwise distinct. If we have $\{z_\ell\}=\{\infty\}$, which corresponds to $\mathfrak{f}(z_\ell)=0$, it follows that $\mathfrak{f}(\mu_i)\neq0,\,\forall\,i\in\{1:\nu\}$. Thus, $S$ is invertible, which yields\vspace{-2mm}
			\begin{equation*}
				\resizebox{\columnwidth}{!}{
                $
                \nu=\mathrm{rank}\begin{bmatrix}
					S&G
				\end{bmatrix}=\mathrm{rank}\begin{bmatrix}
					S&G
				\end{bmatrix}\begin{bmatrix}
					I_\nu & O\\ -L & 1
				\end{bmatrix}=\mathrm{rank}\begin{bmatrix}
					S-GL&G
				\end{bmatrix}.
                $
                }\vspace{-1mm}
			\end{equation*}
			Recalling that \eqref{eq:vpg_3} holds, we now employ Corollary~1 from \cite{gen_ss}, which states that the right-hand implication in \eqref{eq:zero_cond} is true if and only if the following identity is satisfied\vspace{-1mm}
			\begin{equation*}
				\det\scriptsize\begin{bmatrix}
					-(S-GL) & -G \\ C\Pi & D
				\end{bmatrix}=0.\vspace{-1mm}\normalsize
			\end{equation*}
			This is indeed the case, since \eqref{eq:vpg_1} holds and $\mathfrak{f}(z_\ell)=0$.
		}\qed
	
	\bibliographystyle{ieeetr}
	\bibliography{ref}

@article{aug_sparse,
	author={Speril\u{a}, Andrei and Oar\u{a}, Cristian and Ciubotaru, Bogdan and Sab\u{a}u, {\c{S}}erban},
	journal={{IEEE Transactions on Automatic Control}}, 
	title={{Distributed Control of Descriptor Networks: A Convex Procedure for Augmented Sparsity}}, 
	year={2023},
	volume={68},
	number={12},
	pages={8067--8074},
	doi={10.1109/TAC.2023.3301949}
}

@article{BMI2LMI,
	
	author={R. {Doelman} and M. {Verhaegen}},
	
	journal={In Proc. of the 2016 European Control Conference}, 
	
	title={{Sequential convex relaxation for convex optimization with bilinear matrix equalities}}, 
	
	year={2016},
	
	pages={1946--1951},
	
}

@ARTICLE{app,
  author={Hu, Yinlong and Cheng, Changjun and Li, Jia and Chen, Michael Z. Q. and Du, Haiping},
  journal={IEEE Transactions on Industrial Electronics}, 
  title={{Design and Experimental Analysis of an Operational Amplifier Circuit-Based Mechatronic Semiactive Inerter}}, 
  year={2024},
  volume={71},
  number={4},
  pages={3915--3923},
  doi={10.1109/TIE.2023.3273263}
}

@book{gantmacher,
	title={{The Theory of Matrices}},
	author={Gantmacher, Feliks},
	year={1959},
	publisher={American Math. Soc.}
}

@book{Won,
  title={Linear Multivariable Control: A Geometric Approach},
  author={Wonham, W.M.},
  series={Applications of mathematics},
  year={1985},
  publisher={Springer New York}
}

@article{gen_ss,
author = {G. Verghese and P. {Van Dooren} and T. Kailath},
title = {Properties of the system matrix of a generalized state-space system},
journal = {International Journal of Control},
volume = {30},
number = {2},
pages = {235--243},
year = {1979},
publisher = {Taylor \& Francis},
}

@book{zhou,
	title={{Robust and Optimal Control}},
	author={Zhou, Kemin and Doyle, John and Glover, Keith},
	year={1996},
	publisher={Prentice-Hall}
}

@article{H2_ds,
  title={The optimal $\mathcal{H}_2$ controller for generalized discrete-time systems},
  author={Speril{\u{a}}, Andrei and Ciubotaru, Bogdan D and Oar{\u{a}}, Cristian},
  journal={Automatica},
  volume={133},
  number={109889},
  year={2021},
  publisher={Elsevier}
}

@article{H2_ct,
  title={{$\mathcal{H}_2$ Output Feedback Control of Differential-Algebraic Systems}},
  author={Speril{\u{a}}, Andrei and Oar{\u{a}}, Cristian and Ciubotaru, Bogdan D},
  journal={IEEE Control Systems Letters},
  volume={6},
  pages={542--547},
  year={2021},
  publisher={IEEE}
}

@book{complexvar,
	title={Complex Analysis: The Argument Principle In Analysis And Topology},
	author={Beardon, AF},
	publisher={John Wiley \& Sons},
	year={1979}
}

@article{css_orig,
  title={{Minimal Factorization of Rational Matrix Functions}},
  author={M. Rakowsi},
  journal={IEEE Trans. on Circuits and Systems I},
  year={1992},
  volume={39},
  pages={440--445},}

@book{DAE_surv,
  title={{Applications of Differential-Algebraic Equations: Examples and Benchmarks}},
  author={Campbell, Stephen and Ilchmann, Achim and Mehrmann, Volker and Reis, Timo and others},
  year={2019},
  publisher={Springer}
}

@ARTICLE{H2_MM,
  author={Necoar\u{a}, Ion and Ionescu, Tudor C.},
  journal={IEEE Trans. on Automatic Control}, 
  title={{$\mathcal{H}_2$ Model Reduction of Linear Network Systems by Moment Matching and Optimization}}, 
  year={2020},
  volume={65},
  number={12},
  pages={5328--5335},}

@ARTICLE{NL_MM,
  author={Ionescu, Tudor C. and Astolfi, Alessandro},
  journal={IEEE Trans. on Automatic Control}, 
  title={{Nonlinear Moment Matching-Based Model Order Reduction}}, 
  year={2016},
  volume={61},
  number={10},
  pages={2837--2847},}

@ARTICLE{TS_MM,
  author={Ionescu, T. C.},
  journal={IEEE Trans. on Automatic Control}, 
  title={{Two-Sided Time-Domain Moment Matching for Linear Systems}}, 
  year={2016},
  volume={61},
  number={9},
  pages={2632--2637},}

@ARTICLE{ORIG_MM,
  author={Astolfi, Alessandro},
  journal={IEEE Trans. on Automatic Control}, 
  title={{Model Reduction by Moment Matching for Linear and Nonlinear Systems}}, 
  year={2010},
  volume={55},
  number={10},
  pages={2321--2336},}

@article{PZ_MM,
	title={Model reduction with pole-zero placement and high order moment matching},
	author={Ionescu, Tudor C and Iftime, Orest V and I. Necoar\u{a}},
	journal={Automatica},
	volume={138},
	number={110140},
	year={2022},
	publisher={Elsevier}
}

@inproceedings{IO_DD_MM,
	title={Enforcing input-output behavior in data-driven moment matching},
	author={Bhattacharjee, Debraj and Moreschini, Alessio and Astolfi, Alessandro},
	booktitle={2025 IEEE 64th Conference on Decision and Control},
	pages={5806--5811},
	year={2025},
	organization={IEEE}
}

@article{KBL_MM,
	title={{Moment Matching by Kernel-Based Learning}},
	author={Moreschini, Alessio and Scandella, Matteo and Astolfi, Alessandro and Parisini, Thomas},
	journal={IEEE Trans. on Automatic Control},
	year={2025},
	publisher={IEEE}
}

@article{DAE_MM,
	title={Parameterization of all differential-algebraic moment matching interpolants},
	author={Simard, Joel D and Moreschini, Alessio and Astolfi, Alessandro},
	journal={IEEE Trans. on Automatic Control},
	volume={70},
	number={3},
	pages={1875--1882},
	year={2024},
	publisher={IEEE}
}

@inproceedings{DISC_MM,
	title={{Moment matching for linear systems in discrete-time: Towards enhanced performance}},
	author={Bhattacharjee, Debraj and Moreschini, Alessio and Astolfi, Alessandro},
	booktitle={2024 IEEE 63rd Conference on Decision and Control (CDC)},
	pages={7320--7325},
	year={2024},
	organization={IEEE}
}

@article{johnson_fa_di_bruno,
  title={{The Curious History of Fa{\`a} di Bruno's Formula}},
  author={Johnson, Warren P},
  journal={American Mathematical Monthly},
  volume={109},
  number={3},
  pages={217--234},
  year={2002},
  publisher={Taylor \& Francis}
}

@book{matrix_analysis,
  title={Matrix Analysis},
  author={Johnson, Charles R and Horn, Roger A},
  year={2012},
  edition={Second},
  publisher={Cambridge University Press Cambridge}
}

\end{document}